%% file: main.tex
\documentclass[sigconf]{acmart}

\makeatletter
\def\mdseries@tt{m}
\makeatother

\usepackage{paradise}
\usepackage{stmaryrd}
\usepackage{paralist}
\usepackage{rotating}
\usepackage{pbox}
\usepackage{listings}
\usepackage{booktabs}
\usepackage{siunitx}
\usepackage{subcaption}
\usepackage{tabularx}
\usepackage[frozencache]{minted}
\usepackage{array}

\newcommand{\sort}[1]{\ensuremath{[#1]}}

\newcommand{\extract}[2]{\ensuremath{\texttt{extract}^{#1}_{#2}}}

\newcommand{\eval}[1]{\ensuremath{\llbracket #1 \rrbracket}}
\newcommand{\bool}{\ensuremath{\mathit{Boolean}}\xspace}
\newcommand{\true}{\top\xspace}

\newcommand{\free}[1]{\ensuremath{\mathit{free}(#1)}}
\newcommand{\defeq}{\ensuremath{\stackrel{\mathrm{df}}{=}}}
\newcommand{\NP}{\mathsf{NP}}
\newcommand{\PSPACE}{\mathsf{PSPACE}}
\newcommand{\SMT}{\textsc{smt}\xspace}
\newcommand{\SAT}{\textsc{sat}\xspace}
\newcommand{\CEDS}{\textsc{ceds}\xspace}

\newcommand{\pVarsi}[2]{\ensuremath{\mathit{vars}}(#1, #2)}
\newcommand{\pVars}[1]{\ensuremath{\mathit{vars}}(#1)}
\newcommand{\progVars}{\ensuremath{\mathit{pVars}}\xspace}
\newcommand{\prog}{\ensuremath{\mathit{prog}}\xspace}
\newcommand{\gen}{\ensuremath{\mathit{gen}}\xspace}

\newcommand{\notsubseteq}{\ensuremath{\mathit{notsubseteq}}}

\newcolumntype{L}[1]{>{\raggedright\let\newline\\\arraybackslash\hspace{0pt}}m{#1}}
\newcolumntype{C}[1]{>{\centering\let\newline\\\arraybackslash\hspace{0pt}}m{#1}}
\newcolumntype{R}[1]{>{\raggedleft\let\newline\\\arraybackslash\hspace{0pt}}m{#1}}

\begin{document}
\title[SMT Queries Decomposition in Semi-Symbolic Model
Checking]{SMT Queries Decomposition and Caching \\ in Semi-Symbolic
  Model Checking}




\author{Jan Mrázek}
\affiliation{\institution{\fimuni}}
\email{xmrazek7@fi.muni.cz}

\author{Martin Jonáš}
\affiliation{\institution{\fimuni}}
\email{xjonas@fi.muni.cz}

\author{Jiří Barnat}
\affiliation{\institution{\fimuni}}
\email{xbarnat@fi.muni.cz}

\begin{abstract}

In semi-symbolic (control-explicit data-symbolic) model checking the state-space
explosion problem is fought by representing sets of states by first-order
formulas over the bit-vector theory. In this model checking approach, most of
the verification time is spent in an \SMT solver on deciding satisfiability of
quantified queries, which represent equality of symbolic states. In this paper,
we introduce a new scheme for decomposition of symbolic states, which can be
used to significantly improve the performance of any semi-symbolic model
checker. Using the decomposition, a model checker can issue much simpler and
smaller queries to the solver when compared to the original case. Some \SMT calls
may be even avoided completely, as the satisfaction of some of the simplified
formulas can be decided syntactically. Moreover, the decomposition allows for an
efficient caching scheme for quantified formulas. To support our theoretical
contribution, we show the performance gain of our model checker \symdivine
on a set of examples from the Software Verification Competition.

\end{abstract}

\keywords{Model checking, semi-symbolic model checking, state slicing,
  caching, formal verification, \SMT query decomposition, \symdivine.}

\acmDOI{N/A}

\acmISBN{N/A}

\acmConference{Submitted to arxiv.org}{}{}{}
\acmYear{2017}
\setcopyright{none}
\acmPrice{}

\maketitle

\section{Introduction} \label{sec:intro}

Automated formal verification of a real-world code is an ultimate goal
for both academia and industry. One of the methods that are most
suitable for achieving this goal is model checking. Originally, the
model checking approach was designed for verification of distributed
systems that were modeled in some appropriate modeling language for
the purpose of verification. However, recent achievements made model
checking tools more general and applicable directly to source codes of
middle-sized software projects. For example, model checker
\divine~\cite{DiVinE30} allows for a direct model checking of an
unmodified C or C++ code. The general applicability and widespread
use of model checking is unfortunately limited by the well-known state
space explosion problem, i.e. data structures that need to be produced
and explored in order to complete the verification process may blow-up
exponentially with respect to the size of the model-checked source
code.

The exponential growth of data structures comes from two sources --
the interleaving of parallel processes in the system being verified
(control-flow non-determinism) and from processing of input data (data
non-determinism). The control flow non-determinism can be alleviated
by state space reduction methods, e.g. $\tau$-reduction~\cite{RBB13},
or partial order reduction~\cite{Peled1998}. Data non-determinism, on
the other hand, is typically dealt with using abstract
interpretation~\cite{AbsInt}, BDD data structures~\cite{CAV02}, or
using formulas in a suitable logic and \SAT or \SMT
solvers~\cite{NuXmv}. Surprisingly, most model checking tools
available to general public focus mainly on a single type of
non-determinism. To address this issue, a semi-symbolic approach to
model checking has been introduced recently. This approach is also
called \emph{Control-Explicit Data-Symbolic} (\CEDS) model
checking~\cite{BBH14}. \CEDS approach basically follows the
enumerative model checking scheme with the exception that the data
parts of states are represented symbolically, which allows a \CEDS
model checker to represent multiple states with the same control-flow
part as a single compactly represented object -- so-called
\emph{multi-state}. This efficiently mitigates data non-determinism
blow-up during the state space exploration, but introduces costly
operations for working with symbolic parts of states. The \CEDS
approach is implemented, e.g. within the tool
\symdivine~\cite{spin2016}.

\symdivine employs first-order formulas and an \SMT solver to deal
with multiple possible values of symbolic data at one control-flow
location. As a result, most of the verification time of \symdivine
verification process is spent in queries to the \SMT solver. In this
paper, we introduce a new scheme for slicing of states in \CEDS
approach, which reduces complexity of issued \SMT queries and allows
for their caching. Caching of \SMT queries has a significant impact on
the performance of the whole model checking procedure, and it is not
easy to achieve without the state slicing as the queries issued during
verification contain universal quantifiers. Furthermore, the state
slicing allows for further optimizations in the model checking
procedure: in some cases, the equality of symbolic parts of states may
be solved purely syntactically without even calling an \SMT solver,
which brings another performance boost. In the paper, we give the
necessary theoretical background for state slicing, as well as a
report on our implementation in the tool \symdivine and its
experimental evaluation on benchmarks from Software Verification
Competition (SV-COMP)~\cite{SVCOMP}.

The rest of the paper is organized as follows. In
Section~\ref{sec:preliminaries}, we give the necessary introduction to
the \CEDS model checking and context for the state slicing and
caching, which are described in Section~\ref{sec:slicing}. In
Section~\ref{sec:impleval}, we describe details of our implementation
and report on experimental evaluation we performed to measure the
benefit of our new approach. Finally, we conclude the paper in
Section~\ref{sec:conclusion}.

\section{Preliminaries} \label{sec:preliminaries}

The performance of a semi-symbolic model checker relies on a compact
representation of the sets of states and on efficient implementations
of operations using this representation. Although several
representations have been proposed and tested, the first-order
formulas over the theory of fixed size bit-vectors have shown to be
most efficient in practice. Using this representation, tests for
emptiness of a symbolic state and for equivalence of two symbolic
states are performed by queries to an \SMT solver capable of handling
quantified bit-vector formulas.

\subsection{Theory of Fixed Sized Bit-vectors}

In this section, we briefly recall the bit-vector theory, which is
used to represent sets of valuations in the latest versions of the
tool \symdivine.

The \emph{theory of fixed sized bit-vectors} is a many-sorted
first-order theory with infinitely many sorts $\sort{n}$ corresponding
to bit-vectors of length $n$. Additionally, as in
Hadarean~\cite{Had15}, we suppose a distinguished sort $\bool$ and
instead of treating formulas and terms differently, we consider
formulas as merely the terms of sort $\bool$. The only predicate
symbols in the BV theory are $=$, $\leq_u$, and $\leq_s$, representing
equality, unsigned inequality of binary-encoded natural numbers, and
signed inequality of integers in $2$'s complement representation,
respectively. Functions symbols in the theory are
$+, \times, \div, \&, \mid, \oplus, \ll, \gg, \cdot, \extract{n}{p}$,
representing addition, multiplication, unsigned division, bit-wise
and, bit-wise or, bit-wise exclusive or, left-shift, right-shift,
concatenation, and extraction of $n$ bits starting from position $p$,
respectively. For the map $\mu$ assigning to each variable a value in
a domain of its sort, we denote as $\eval{\_}_\mu$ the evaluation
function, which to each formula $\varphi$ assigns the value
$\eval{\varphi}_\mu$. This value is obtained by substituting free
variables in $\varphi$ by values given by $\mu$ and evaluating all
functions, predicates, and quantifiers according to their standard
interpretation. The formula $\varphi$ is \emph{satisfiable} if
$\eval{\varphi}_\mu = \true$ for some mapping $\mu$; it is
\emph{unsatisfiable} otherwise. Formulas $\varphi$ and $\psi$ with the
same set of free variables are \emph{equivalent} if
$\eval{\varphi}_\mu = \eval{\psi}_\mu$ for all assignments $\mu$.
Further, formulas $\varphi$ and $\psi$ are \emph{equisatisfiable} if
both are satisfiable or both are unsatisfiable. If $\Phi$ is a finite
set of formulas, we denote as $\bigwedge \Phi$ the conjunction of all
formulas in $\Phi$. A set of free variables of the formula $\varphi$
is defined as usual and denoted $\free{\varphi}$. Formulas $\varphi$
and $\psi$ are called \emph{(syntactically) dependent} if they do not
share any free variable, i.e.
$\free{\varphi} \cap \free{\psi} = \emptyset$. The precise description
of the many-sorted logic can be found for example in Barrett et
al.~\cite{BSST09}. For a precise description of the syntax and
semantics of the bit-vector theory, we refer the reader to
Hadarean~\cite{Had15}.

\subsection{\symdivine} \label{sec:symdivine}

\symdivine is a semi-symbolic model checker aiming for verification of
real world C and C++ programs featuring parallelism. To achieve
precise semantics of the input languages and, at the same time, to
ease the parsing, it is built upon the \llvm compiler framework. In
order to verify real-world pieces of code, \symdivine provides
intrinsic implementations of a subset of the pthread library to allow
parallelism and also of a subset of SV-COMP interface to allow users
to model non-deterministic inputs in their programs. Internally,
\symdivine relies on the \CEDS approach, in detail described in
\cite{BBH14}. The \CEDS approach allows \symdivine to verify both
safety and \ltl properties of programs under inspection. In the
following subsection, we cover the basics and details relevant to this
paper. For further information, we kindly refer the reader to the
original paper.

\subsection{Control-Explicit Data-Symbolic Model Checking}\label{subsec:ceds}

In the standard explicit state model checking, the state space graph
of a program is explored by an exhaustive enumeration of its states,
until an error is found or all reachable states have been enumerated.
\symdivine basically follows the same idea, but instead of enumerating
states for all possible input values, it employs a \CEDS approach in
which the inputs of the program are treated in a symbolic
manner. While a purely explicit-state model checker has to produce a
new state for each and every possible input value, in \symdivine a set
of states that differ only in data values may be represented with a
single data structure, the so-called \emph {multi-state}.

A multi-state consists of an explicit control location and a set of
program's memory valuations. In practice, a set of memory valuations is
usually not listed explicitly, but in a more succinct
representation. By providing procedures for deciding whether the set of
memory valuations is empty and whether two sets of memory valuations
are equal, we can easily mimic most of explicit-state model checking
algorithms~\cite{BBH14} -- from a simple reachability of error states
to the full \ltl model checking. By operating on multi-states,
\symdivine can achieve up to exponential time and memory savings,
compared to purely explicit approaches.

Although the \CEDS approach is independent of the multi-state
representation, the choice of the representation can have an enormous
effect on the verification performance. Most recent versions of
\symdivine use a first-order formula over the theory of fixed size
bit-vectors to represent a multi-state. In this representation, the
set of represented memory valuations is precisely the set of satisfying
assignments to the given formula.

\subsection{Multi-state Representation}

Formally, a multi-state in \CEDS approach is a tuple $(c, m, \varphi)$,
where $c$ is an explicit control part, $m$ is an explicit memory shape
and $\varphi$ is a quantifier-free first-order formula over the theory of
fixed size bit-vectors.

A \emph{control part} $c$ is a tuple of call stacks, which contains
contents of stack for each active thread of a program under
inspection. Note that not all multi-states in a state space have to
contain the same number of call stacks, as threads can be spawned and
killed during the program execution. For each thread, a call stack is
composed of frames corresponding to function calls, which hold a
program counter and a reference to the segment in the
memory shape $m$ that corresponds to the active function call.

A \emph{memory shape} is a collection of segments, where each segment
has a unique identifier and contains a list of corresponding
variables. Each variable in a segment has associated several pieces of
information:
\begin{enumerate}
\item a type determining whether the variable is a pointer or a value,
  and if it is a value, also its bit-width,
\item explicit/symbolic mark, and
\item a value in case of variables marked as explicit.
\end{enumerate}
Given this setup, each variable can be uniquely identified by a pair
$(s,p)$, where $s$ is a segment identifier and $p$ is the position of
the variable inside that segment.

Note that \symdivine distinguishes between the control flow and a
memory shape, because \llvm instructions like \texttt{alloca} may
allocate memory that can escape from the function's segment. This
happens for example when a function obtains an argument that is a
value that was obtained by \texttt{alloca}. Therefore a single
multi-state can contain more segments than frames of call stacks.

Finally, each multi-state contains a quantifier-free formula $\varphi$
that represents possible values of all variables marked as symbolic in
the memory shape. Because a single program variable can be assigned to
multiple times, the formula $\varphi$ may for each program variable
$(s,p)$ contain multiple variables of form $(s, p)^{\mathit{gen}}$, where
$\mathit{gen}\in\mathbb{N}$. Variable $(s, p)^{\mathit{gen}}$ represents a value of the
program variable $(s,p)$ just after $\mathit{gen}$-th assignment to that
variable. The number $\mathit{gen}$ is called the \emph{generation} of
the variable $(s, p)^{\mathit{gen}}$. Let $\prog(s,p)$ denote the variable
$(s,p)^{\mathit{lgen}}$, where $lgen$ is the greatest generation of all
variables $(s,p)^i$ in $\varphi$. The variable $\prog(s,p)$
intuitively represents a real value of the program variable $(s,p)$ in
the multi-state. Therefore, given a model $\mu$ of a formula
$\varphi$, we can obtain a possible valuation of program variables in
a multi-state by restricting $\mu$ only to variables of form
$\prog(s,p)$. Thus a single satisfiability query to an \SMT
solver is sufficient to determine whether the set of states represented by
a multi-state is empty or not. This query is called the
\emph{emptiness check} for a multi-state.

In the further text, we refer to the control part $c$ and the memory
shape $m$ in a multi-state $s$ as the \emph{explicit part of
  $s$}. Similarly, we refer to $\varphi$ as the \emph{symbolic part of
  $s$}. Furthermore, if the segment identifier, position of the
variable in the segment and the generation of the variable are not
important, we refer to variables of formula $\varphi$ only as
$x,y,z,a,b, \ldots$. For the convenience, we suppose that each program
variable defined in the program location $c$ has at least one
corresponding variable in the formula $\varphi$. This assumption is
without the loss of generality, as for each program variable $(s,p)$,
a vacuous equality $(s,p)^1 = (s,p)^1$ can be conjoined to the formula
$\varphi$.

During the interpretation of the program, the verifier has to be able
to compute all successors of a given multi-state in order to construct the
complete state-space graph. Successors of a node can arise by two
types of operations:
\begin{itemize}
\item transformation caused by arithmetic, bitwise and memory instructions and
\item pruning caused by control-flow branching and by atomic propositions.
\end{itemize}
Both of these operations can be modeled by changing the formula
$\varphi$. For a given formula $\varphi$ and a program variable
$(s,p)$, denote as $\gen(s,p)$ the number $i$ such that
$\prog(s,p) = (s,p)^i$; this number is called the \emph{last
  generation of $(s,p)$ in $\varphi$}. Suppose we want to compute a
successor of a multi-state with a symbolic part $\varphi$. Then the
state resulting from a program instruction
$(s,p) = (s_2,p_2) \oplus (s_3,p_3)$, where $\oplus$ is a binary
arithmetic, bitwise or memory instruction, has a symbolic part
\[
  \varphi' \defeq \varphi~\wedge~((s,p)^{\gen(s,p) + 1} = \prog(s_2,p_2) \mathbin{\hat{\oplus}} \prog(s_3,p_3)),
\]
where $\hat{\oplus}$ is the corresponding function symbol in the
bit-vector theory.  Similarly, pruning a multi-state by a binary
predicate $(s_1,p_1) \boxtimes (s_2,p_2)$ results in a multi-state
with the symbolic part
\[
  \varphi' \defeq \varphi~\wedge~ (\prog(s_1,p_1) \mathbin{\hat{\boxtimes}} \prog(s_2,p_2)),
\]
where $\hat{\boxtimes}$ is again the corresponding predicate symbol in
the bit-vector theory.

\begin{example}
  Consider the following single-threaded C program.
  \begin{minted}[linenos]{c}
int main() {
    int x = nondet()
    int y = x + 5
    int x = x + 10
    if (x > y)
        y = y + 1
}
\end{minted}
  Control parts of all states are straightforward, since they contain
  only one stack with the associated program counter and a single
  memory segment. We describe the memory shape and the symbolic part
  of the multi-state $s$ that represents the state of the program on
  the end of the line 6. The single memory segment is labeled by 1
  and contains two variables: \texttt{x} labeled by the index 1 and
  \texttt{y} labeled by the index 2. Both these variables are marked
  as symbolic values. Therefore, the symbolic part contains variables
  $(1,1)^i$, which represent values of the program variable $x$, and
  variables $(1,2)^i$, which represent values of the program variable
  $y$. For the sake of readability, we will refer to variables
  $(1,1)^i$ as $x^i$ and to variables $(1,2)^i$ as $y^i$. In
  particular, the symbolic part $\varphi$ of this multi-state is
\[
  (y^1 = x^1 + 5)~\wedge~(x^2 = x^1 + 10)~\wedge~(x^2 \leq_s y^1)~\wedge~(y^2 = y^1 + 1).
\]
\end{example}

\subsection{Multi-state Equality Check}
We now describe how an \SMT solver can be used to decide whether two
multi-states represent the same set of concrete states. We further
refer to this check as to the \emph{equivalence check}.

Let $s_1$ and $s_2$ be multi-states with the same explicit part. That
is, $s_1 = (c, m, \varphi)$ and $s_2 = (c, m, \psi)$ for a control
part $c$, memory shape $m$ and formulas $\varphi$ and $\psi$. Let
$\free{\varphi} = \{x_1, \ldots, x_n \}$ and
$\free{\psi} = \{y_1, \ldots, y_m \}$. Furthermore, let us denote the
set of program variables defined at the control location $c$ as
$\pVars{c}$. For each program variable $p$, there is a variable $x^p$
in $\varphi$ that represents the last generation of the program
variable $p$ in the multi-state $s_1$.  Analogously, the last
generation of the program variable $p$ in $s_2$ is represented by a
variable $y^p$ in $\psi$.

We want to decide whether the sets of states represented by $s_1$ and $s_2$
are equal. To determine this, we define a formula
$\notsubseteq(s_1, s_2)$, which is satisfiable precisely if there is a
state represented by $s_1$ that is not represented by $s_2$:
\[
  \notsubseteq(s_1, s_2) \defeq \varphi~\land~ \forall y_1 \ldots y_m
  \, \Big ( \psi \Rightarrow \bigvee_{p \in \pVars{c}} (x^p \not =
  y^p) \Big)
\]
The equality of two multi-states can now be determined by using an \SMT
solver: the states $s_1$ and $s_2$ are equal precisely if both of
formulas $\notsubseteq(s_1, s_2)$ and $\notsubseteq(s_2, s_1)$ are
unsatisfiable, i.e. there is no memory valuation that is represented
only by one of the multi-states. However, the equality check requires
a quantified \SMT query, which is usually more expensive than the
quantifier-free one. For example, in the theory of fixed-size
bit-vectors, deciding satisfiability of a quantifier-free formula is
$\NP$-complete, whereas deciding satisfiability of a formula with
quantifiers is $\PSPACE$-complete~\cite{KFB16}.

\section{State Slicing and Caching} \label{sec:slicing}

In order to reduce the cost of a quantified \SMT query, we observe that
real-world programs often contain lots of independent variables, i.e. pairs of
variables such that a change in any of them does not affect the other. We give
two simple examples that illustrate this phenomenon: a sequential program in Figure
\ref{code:sec} and a multi-threaded program in Figure \ref{code:par}.

\begin{figure}[htb!]
    \begin{minted}[linenos]{c}
int foo( int a, int b ) {
    int result;
    // Store a complicated expression,
    // e.g. modular arithmetics, to result
    return result;
}

int main() {
    int x = foo( nondet(), nondet() );
    int y = 0;
    for ( uint n = nondet(); n % 42 != 0; n++ ) {
        y++;
    }
    return x * y;
}
    \end{minted}
    \caption{Sequential code demonstrating the motivation for the state slicing. }
    \label{code:sec}
\end{figure}

Let us examine state spaces of these two examples. To keep the
explanation simple, we provide a rather high-level description and omit
technical details of the real implementation of \symdivine that
operates on top of the \llvm infrastructure. The example in Figure
\ref{code:sec} consists of two functions -- \texttt{main} and
\texttt{foo}. The function \texttt{foo} represents a function
taking two integer arguments and computing an integer result. Effect
of this function can be described by a formula $\psi(a,b,x)$ as a
relation between input arguments and the return value.

When this example is examined by \symdivine, among others, the following states are produced:
\begin{itemize}
\item an initial state (before \texttt{main} starts) -- $s_{\mathit{init}}$;
\item a state $s_i$ after every cycle iteration;
\item a final state $s_{\mathit{final}_i}$ corresponding to the
  location just after line 14 for each number $0 \leq i \leq 42$ of
  the performed iterations of the for loop.
\end{itemize}
Note that in this example the cycle has to be unrolled for each
iteration. This is due to the presence of a variable \texttt{y}, which
has different value in every iteration and therefore each iteration
can be distinguished from every other.

The symbolic parts of these states are:
\begin{align*}
  s_{\mathit{init}} = {}&\true, \\
  s_i = {}&\psi(a,b,x^1) \wedge y^1=0 \wedge {} \\
          &n^1  \bmod 42 \neq 0~\wedge~y^2 = y^1 + 1\wedge~n^2=n^1+1~\wedge{} \\
          &\cdots \\
          &n^i  \bmod 42 \neq 0~\wedge~y^{i+1}=y^i+1~\wedge~n^{i+1} = n^i + 1 \\
  s_{\mathit{final}_i} = {}&\psi(a,b,x_1) \wedge y^1=0 \wedge {} \\
          &n^1  \bmod 42 \neq 0~\wedge~y^2=y^1+1~\wedge~n^2 = n^1 + 1~\wedge {} \\
          &\cdots \\
          &n^i  \bmod 42 \neq 0~\wedge~y^{i+1}=y^i+1~\wedge~n^{i+1} = n^i + 1 \wedge {} \\
          &n^{i+1} \bmod 42 = 0~\wedge~\mathit{returnValue}=x^1 \times y^{i+1}.
\end{align*}

During the generation of the state space, \symdivine tries to merge
newly generated states with already existing states and thus performs
the equality check. As merging occurs only on the states with the same
control-flow locations, equality checks are issued whenever a state
$s_i$ or $s_{\mathit{final}_i}$ is produced. If we examine the
formulas that are checked for satisfiability during equality checks
for $s_i$, we can observe that although the values of variables $x, a$
and $b$ do not change during the loop, all queries test their
equality. This forces the \SMT solver to consider the expensive
formula $\psi(a,b,x_0)$ for each unrolling of the cycle and thus slows
the verification of the program down. This effect is even stronger if
there is a large number of states per a control-flow location. The
observed pattern of computing multiple independent sub-results in
advance and combining them later in the computation of the program is
quite common in sequential programs.

\begin{figure}[htb!]
    \begin{minted}[linenos]{c}
volatile int x;
volatile int y;

void foo( arg ) {
    arg* = nondet();
    while( *arg % 5 ) {
        ( *arg )++;
    }
}

int main() {
    t1 = new_thread( foo, &x );
    t2 = new_thread( foo, &y );
    join( t1 );
    join( t2 );
    return x + y;
}

    \end{minted}
    \caption{Parallel code demonstrating the motivation for the state slicing. }
    \label{code:par}
\end{figure}

Similarly, if we examine the example of a parallel code in Figure
\ref{code:par}, we find out that the \SMT solver again does more work
than necessary. If the program under inspection contains threads, they
do not necessarily interact in each step of the computation. Some
thread interleavings are not interesting from verification point of
view and therefore \symdivine implements state-space reduction
techniques. However, these techniques do not eliminate all states
produced by equivalent interleaving of threads. By using a similar
approach as in the sequential case, i.e. dividing the formula into
unrelated parts, we can decrease the verification time and alleviate
$\tau$-reduction's imperfections. In this case, the independent parts
of the formula consist of variables local to each of the
threads.


These observations motivate the decomposition of multi-states into
independent parts, in which each group of independent variables is
represented by one first-order formula. We show that such
representations can decrease size and in most of the cases also the
number of \SMT queries necessary during the program verification. By
decomposing multi-states into multiple independent parts, emptiness
and equality checks can be performed independently on each of these
parts. The benefit of thus modified emptiness checks is
twofold. First, although the number of performed \SMT queries grows,
they are much simpler, as the \SMT solver does not have to reason
about independent variables, which have no effect on the
satisfiability. In many cases, such queries can be decided by purely
syntactic decision procedures, without even using the \SMT
solver. Second, when an operation is performed on a subset of program
variables, the unrelated part of the multi-state does not change, and
the resulting queries can thus be efficiently cached.

\subsection{Sliced Multi-states}

A \emph{sliced multi-state} is a triple
$(c, m, \{ \varphi_i \}_{1 \leq i \leq k})$, where $c$ and $m$ are as
before and $\varphi_i$ are mutually independent formulas. Formulas
$\varphi_i$ are called \emph{independent symbolic parts} of the
state. Intuitively, each $\varphi_i$ describes possible memory
valuations of a set of independent program variables in a given
control location. Semantics of sliced multi-states is straightforward
-- a set of concrete program states represented by a sliced
multi-state $(c, m, \{ \varphi_i \}_{1 \leq i \leq k})$ is defined as
a set of program states represented by the (ordinary) multi-state
$(c, m, \bigwedge_{1 \leq i \leq k} \varphi_i)$. We say that an ordinary
multi-state $s_1 = (c_1, m_1, \varphi)$ is syntactically equivalent to
the sliced multi-state
$s_2 = (c_2, m_2, \{ \varphi_i \}_{1 \leq i \leq k})$ if $c_1 = c_2$,
$m_1 = m_2$, and $\varphi$ is equal to
$\bigwedge_{1 \leq i \leq k} \varphi_i$ up to the ordering of the
conjuncts.

Although each multi-state $s$ can be converted to a syntactically
equivalent sliced multi-state in many ways, there always exists a
syntactically equivalent sliced multi-state with the largest number of
independent symbolic parts. Recall that each multi-state created during the
interpretation of the program is of form
$(c, m, \bigwedge_{1 \leq i \leq k} \varphi_i)$, with possible
dependencies among formulas $\varphi_i$. Let
$\Phi = \{ \Psi_1, \ldots, \Psi_l \}$ be the set of equivalence
classes of $\{ \varphi_i \}_{1 \leq i \leq k}$ modulo the dependence
relation. A sliced multi-state syntactically equivalent to $s$ with
the largest number of independent symbolic parts is then
$(c, m, \{ \bigwedge \Psi_i \mid 1 \leq i \leq l \})$.

\begin{example}
  \label{slicedMaxGroups}
  Let
  \[
    (c, m, x = y + z~\wedge~b > c~\wedge~z = a~\wedge~d > 0)
  \] be a
  multi-state. The syntactically equivalent sliced multi-state with
  the largest number of independent symbolic parts is
  \[
    (c, m, \{ (x = y + z \wedge z = a),~(b > c),~(d > 0) \}).
  \]
\end{example}

As before, during the program interpretation a sliced multi-state has
to be transformed by conjoining a formula to represent variable
assignments or program branching. However, the situation is now more
complex, as parts of a sliced multi-state may have to be merged during
the computation, because a conjoined formula can introduce new
dependencies among the variables and parts of the multi-state may have
to be merged together in order for the new sliced multi-state to be
correct. In particular, suppose that
$s = (c, m, \{ \varphi_i \}_{1 \leq i \leq k})$ is a sliced
multi-state and $\psi$ is the formula to be conjoined to it. We partition
the set $\{ \varphi_i \mid 1 \leq i \leq k \}$ to two sets $\Phi$ and
$\Psi$ such that all formulas in $\Phi$ are independent on $\psi$ and
all formulas in $\Psi$ are dependent on $\psi$. Then a result of
conjoining $\psi$ to the state $s$ is the sliced multi-state
$(c, m, (\psi \wedge \bigwedge \Psi) \cup \Phi)$.

\begin{example}
  Consider the sliced multi-state from
  Example~\ref{slicedMaxGroups}. After pruning caused by
  interpretation of the branching $x = b$, the new sliced multi-state is
  \[
    (c, m, \{ (x = b~\wedge~x = y + z~\wedge~z = a~\wedge~b > c),~(d > 0) \}),
  \]
  because symbolic parts $x = y + z \wedge z = a$ and $b > c$ have to
  be merged.
\end{example}

Moreover, to be able to compare two sliced multi-states, their
independent parts have to be in one-to-one correspondence that
respects program variables. To define this formally, we
introduce a function \progVars that for a given formula $\varphi$
returns the set of program variables represented by the formula
$\varphi$, i.e.
$\progVars(\varphi) = \{ (s,p) \mid (s, p)^i \in \free{\varphi} \text{ for
  some } i \in \mathbb{N} \}$.

Two sliced multi-states
$s_1 = (c_1, m_1, \{ \varphi_i \}_{1 \leq i \leq k})$ and
$s_2 = (c_2, m_2, \{ \psi_i \}_{1 \leq i \leq k})$ with the same
number of independent symbolic parts are then said to be
\emph{matching} if

\begin{itemize}
\item they have the same control part,
\item sets
$\progVars(\varphi_i)$ and $\progVars(\varphi_j)$ are disjoint for all $i \not = j$,
\item sets $\progVars(\psi_i)$ and
$\progVars(\psi_j)$ are disjoint for all $i \not = j$, and
\item there is a bijection
$f \colon \{ \varphi_i \mid {1 \leq i \leq k} \} \rightarrow \{ \psi_i
\mid {1 \leq i \leq k} \}$ such that
\[
  \progVars(\varphi_i) = \progVars(f(\varphi_i)).
\]
\end{itemize}

It is easy to see that each two sliced multi-states
$s_1 = (c, m, \Phi)$ and $s_2 = (c, m, \Psi)$ with the same control
part can be transformed to equivalent sliced multi-states that are
matching: an equivalent pair of matching sliced multi-states can
always be obtained by setting $s_1' = (c, m, \{ \bigwedge \Phi \})$
and $s_2' = (c, m, \{ \bigwedge \Psi \})$. However, it is often
possible to obtain equivalent matching sliced multi-states with
more independent groups than one, as the following example
shows.

\begin{example}
  Let $s_1 = (c, m, \Phi)$ and $s_2 = (c, m, \Psi)$, where $c$ and m
  contain only one thread and one segment with variables
  $x,y,z,u,v$. Because a segment number is not important for this
  example, let us denote $i$-th generation of the variable $x$ as
  $x^i$. Let
  \[
    \Phi = \{ (x^1 = y^2 \wedge y^2 = y^1 + 1),~~(z^1 \geq 0),~~(u^1 \leq v^1),~~(u^2 = 5)\}
  \]
  and
  \[
    \Psi = \{ (y^2 = y^1 + u^1),~~(z^1 \geq 3),~~(z^2 = 5),~~(x^1 \leq v^1)\}.
  \]
  Multi-states $s_1$ and $s_2$ are obviously not matching, although
  they have the same number of independent groups. Matching
  multi-states $s_1'$ and $s_2'$ that are equivalent with $s_1$ and
  $s_2$ are for example $s_1' = (c, m, \Phi')$ and
  $s_2' = (c, m, \Psi')$, where
    \[
    \Phi = \{ (x^1 = y^2 \wedge y^2 = y^1 + 1 \wedge u^1 \leq v^1 \wedge u^2 = 5),~~(z^1 \geq 0)\}
  \]
  and
  \[
    \Psi = \{ (y^2 = y^1 + u^1 \wedge x^1 \leq v^1),~~(z^1 \geq 3 \wedge z^2 = 5)\}.
  \]
\end{example}

\subsection{Equality Check}

We can raise the equality check to sliced multi-states. Let $s_1$ and
$s_2$ be matching sliced multi-states with the same control part. That
is, $s_1 = (c, m, \{ \varphi_i \}_{1 \leq i \leq k})$ and
$s_2 = (c, m, \{ \psi_i \}_{1 \leq i \leq k})$ for a control part $c$
and formulas $\varphi_i$ and $\psi_i$ for each $1 \leq i \leq
k$. Further, let $f$ be the function from the definition of matching
sliced multi-states. An obvious way to check equality of these sliced
multi-states is to use the equality procedure for ordinary multi-states
with multi-states $(c, m, \bigwedge_{1 \leq i \leq k}\varphi_i)$ and
$(c, m, \bigwedge_{1 \leq i \leq k}\psi_i)$. However, this way of
performing the equality check completely ignores the structure of
split multi-states. Instead, we present an equality check procedure that
leverages the structure of sliced multi-states and allows simpler
queries that can be easier to solve and also cached efficiently.

As $s_1$ is a sliced multi-state, sets $\free{\varphi_i}$ and
$\free{\varphi_j}$ are disjoint for each $i \not = j$. For each
program variable $p$ defined in the control location $c$, there is a
variable $x^p$ in some formula $\varphi_i$. The variable $x^p$
represents the last generation of the program variable $p$ in the
multi-state $s_1$. Moreover, as the sliced multi-states are matching,
the variable $y^p$, which corresponds to the same program variable, is
a free variable of $f(\varphi_i)$. We without loss of generality
suppose that $\psi_i = f(\varphi_i)$. Note that from the definition of
the function $f$, the formula $\psi_i = f(\varphi_i)$ is the same for all free
variables $x^p$ in $\varphi_i$. Therefore, we can denote as
$\pVarsi{c}{i}$ the set of program variables, whose last generations
are in $\varphi_i$, which is equivalent to the set of program
variables whose last generation is in $\psi_i$.

We can now define for each $i$ the formula
$\notsubseteq(s_1, s_2, i)$, which is satisfiable precisely if there
is a valuation of variables from $\pVarsi{c, m}{i}$ that is present in
the multi-state $s_1$ but not present in the multi-state $s_2$. If
$y_1, \ldots, y_m$ are all free variables that occur in any of the
formulas $\psi_i$, the formula $\notsubseteq$ is defined as follows:
\[
  \notsubseteq(s_1, s_2, i) \defeq \varphi_i~\land~ \forall y_1 \ldots
  y_m \, \Big( \psi_i \Rightarrow \bigvee_{p \in \pVarsi{c}{i}} (x^p
  \not = y^p) \Big)
\]
Using this formula, the equality of the two split multi-states
can be decided by calling an \SMT solver on each of the formulas
$\notsubseteq(s_1, s_2, i)$ independently. We show that the original
formula $\notsubseteq(s_1, s_2)$ is equisatisfiable with the formula
$\bigvee_{1 \leq i \leq k} \notsubseteq(s_1, s_2, i)$ and therefore
the original \SMT query can be replaced by multiple independent \SMT
queries. Although this is true only in the case in which both
multi-states $s_1$ and $s_2$ are non-empty, this is not a problem, as
the equality is checked only for non-empty multi-states.

First, we prove two lemmas that are used in the theorem that states the
correctness of the approach.

\begin{lemma}
  \label{lem:independentSwap}
  Let $\varphi$ be a satisfiable formula, $\rho$ a formula independent
  of $\varphi$ and $\psi$ an arbitrary formula. Then formulas
  $\varphi \wedge (\psi \vee \rho)$ and
  $(\varphi \wedge \psi) \vee \rho$ are equisatisfiable.
\end{lemma}

\begin{proof}
  The implication from left to right follows easily from distributivity
  of conjunction and disjunction. For the converse, suppose
  $(\varphi \wedge \psi) \vee \rho$ is satisfiable and $\mu$ is its
  model.
  If $\mu$ is a model of $(\varphi \wedge \psi)$, it is obviously also
  a model of $\varphi \wedge (\psi \vee \rho)$. Suppose on the other
  hand that $\mu$ is a model of $\rho$ and let $\mu'$ be the
  restriction of $\mu$ to the free variables of $\rho$. As $\varphi$
  is satisfiable, it has a model $\nu$. Then by the independence of
  $\varphi$ and $\rho$, the map $\mu' \cup \nu$ is well-defined and is
  a model of both $\varphi$ and $\rho$ and therefore also of
  $\varphi \wedge (\psi \vee \rho)$.
\end{proof}

\begin{lemma}
  \label{lem:equisatisfiability}
  For non-empty states $s_1$ and $s_2$, formulas
  $\notsubseteq(s_1, s_2)$ and
  $\bigvee_{1 \leq i \leq k} \notsubseteq(s_1, s_2, i)$ are
  equisatisfiable.
\end{lemma}

\begin{proof}
  First, we show that the formula $\notsubseteq(s_1, s_2)$ is
  logically equivalent to the formula
  \[
    \bigwedge_{1 \leq i \leq k} \varphi_i~\land~
    \bigvee_{1 \leq i \leq k} \Big ( \forall y_1 \ldots y_m \, \big ( \neg \psi_i \lor \bigvee_{p \in \pVarsi{c}{i}} (x^p \not = y^p) \big ) \Big )  \\
  \]
  and subsequently we show that this formula is equisatisfiable with
  $\bigvee_{1 \leq i \leq k} \notsubseteq(s_1, s_2, i)$.

  To show the logical equivalence above, we present a sequence of
  formulas that are logically equivalent due to the definitions used and
  well-known first-order tautologies.

  \begin{align}
    \varphi~&\land~
      \forall y_1 \ldots y_m \, \Big ( \psi \Rightarrow \bigvee_{p \in \pVars{c}} (x^p \not = y^p) )  \\
    \bigwedge_{1 \leq i \leq k} \varphi_i~&\land~
      \forall y_1 \ldots y_m \, \Big ( \bigwedge_{1 \leq i \leq k} \psi_i \Rightarrow \bigvee_{\substack{1 \leq i \leq k \\ p \in \pVarsi{c}{i}}} (x^p \not = y^p) \Big )  \\
    \bigwedge_{1 \leq i \leq k} \varphi_i~&\land~
      \forall y_1 \ldots y_m \, \Big ( \bigvee_{1 \leq i \leq k} \neg \psi_i \lor \bigvee_{\substack{1 \leq i \leq k \\ p \in \pVarsi{c}{i}}} (x^p \not = y^p) \Big )  \\
    \bigwedge_{1 \leq i \leq k} \varphi_i~&\land~
      \forall y_1 \ldots y_m \, \Big ( \bigvee_{1 \leq i \leq k} \big ( \neg \psi_i \lor \bigvee_{p \in \pVarsi{c}{i}} (x^p \not = y^p) \big ) \Big )  \\
    \bigwedge_{1 \leq i \leq k} \varphi_i~&\land~
      \bigvee_{1 \leq i \leq k}  \Big ( \forall y_1 \ldots y_m \, ( \neg \psi_i \lor \bigvee_{p \in \pVarsi{c}{i}} (x^p \not = y^p) )  \Big )
  \end{align}
  Some of these steps require explanation. The equivalence of (1) and
  (2) follows from definitions used and reordering of disjuncts. The
  equivalence of (3) and (4) follows again by reordering the
  disjuncts, and the equivalence of (4) and (5) follows from the fact
  that for independent formulas $\varphi$ and $\psi$, the formula
  $\forall x \, (\varphi \vee \psi)$ is equivalent to
  $(\forall x \, \varphi) \vee (\forall x \,
  \psi)$.

  Now, as each $\varphi_i$ is satisfiable and independent on
  $\forall y_1 \ldots y_m \, ( \neg \psi_j \lor \bigvee_{p \in
    \pVarsi{c}{j}} (x^p \not = y^p) ))$ for all $j \not = i$, a
  repeated application of Lemma~\ref{lem:independentSwap} shows that the last formula is
  equisatisfiable with the formula
  \[
    \bigvee_{1 \leq i \leq k} \Big ( \varphi_i~\wedge~\forall y_1 \ldots y_m
    \, ( \neg \psi_i \lor \bigvee_{p \in \pVarsi{c}{i}} (x^p \not =
    y^p) ) \Big ).
  \]
  However, this is precisely the formula
  $\bigvee_{1 \leq i \leq k} \notsubseteq(s_1, s_2, i)$, as was
  required.
\end{proof}

The previous lemma now implies correctness of the improved equality
check.

\begin{theorem}[Correctness]
  For non-empty states $s_1$ and $s_2$, the formula
  $\notsubseteq(s_1, s_2)$ is satisfiable if and only if at least one
  of formulas $\notsubseteq(s_1, s_2, i)$ is satisfiable.
\end{theorem}

\begin{proof}
  Follows from Lemma~\ref{lem:equisatisfiability} by distributivity of existential quantifier
  over disjunctions.
\end{proof}

\begin{example}
  Consider again the program in Figure~\ref{code:sec}. After the state
  decomposition, the independent symbolic parts of split multi-states are
  \begin{align*}
    s_{\mathit{init}} = {}\{ &\true \}, \\
    s_i = {}\{ &\psi(a,b,x_1), \\
                          &(y^1=0 \wedge {} \\
                          &n^1  \bmod 42 \neq 0~\wedge~y^2 = y^1 + 1\wedge~n^2=n^1+1~\wedge{} \\
                          &\cdots \\
                          &n^i  \bmod 42 \neq 0~\wedge~y^{i+1}=y^i+1~\wedge~n^{i+1} = n^i + 1) \} \\
    s_{\mathit{final}_i} = {}\{ &\psi(a,b,x_1) \wedge y^1=0 \wedge {} \\
                          &n^1  \bmod 42 \neq 0~\wedge~y^2=y^1+1~\wedge~n^2 = n^1 + 1~\wedge {} \\
                          &\cdots \\
                          &n^i  \bmod 42 \neq 0~\wedge~y^{i+1}=y^i+1~\wedge~n^{i+1} = n^i + 1 \wedge {} \\
                          &n^{i+1} \bmod 42 = 0~\wedge~\mathit{returnValue}=x^1 \times y^{i+1} \}.
\end{align*}
  The equality check of two states $s_k$, $s_l$ for $k \not = l$ is
  decomposed into two queries: a query concerning the values of
  variables $a,b,x$ and a query concerning variables $y$ and
  $n$. The former query can be decided without a call to an \SMT
  solver, since both states $s_k$ and $s_l$ have the identical
  symbolic part $\psi(a,b,x_1)$ and a simple syntactic check is
  sufficient to determine this. Only the latter query needs to be sent
  to the \SMT solver, which does not have to reason about the
  potentially hard subformula $\psi(a,b,x_1)$ to decide that the
  states $s_k$ and $s_l$ are not equivalent.
\end{example}

\section{Implementation and Evaluation} \label{sec:impleval}

\begin{table*}[t!]
    \centering
    \caption{Summary results showing effects of PartialStore and caching.
    \emph{Solved} marks benchmark count, which the configuration solved in the
    time limit. \emph{Time} is a sum of benchmarks time, which were solved by
    all the configurations.}\label{tab:summary}
    \bigskip
    \setlength{\tabcolsep}{3pt}
    \begin{tabular}{lp{0.3cm}rrp{0.3cm}rrp{0.3cm}rrp{0.3cm}rr}
      \toprule
         &&
        \multicolumn{2}{c}{\pbox{20cm}{\relax\ifvmode\centering\fi SMTStore \\ no cache}} &&
        \multicolumn{2}{c}{\pbox{20cm}{\relax\ifvmode\centering\fi SMTStore \\ with cache}} &&
        \multicolumn{2}{c}{\pbox{20cm}{\relax\ifvmode\centering\fi PartialStore \\ no cache}} &&
        \multicolumn{2}{c}{\pbox{20cm}{\relax\ifvmode\centering\fi PartialStore \\ with cache}} \\
        \addlinespace

      \cmidrule{3-4}
      \cmidrule{6-7}
      \cmidrule{9-10}
      \cmidrule{12-13}
      Category && time[s] & solved && time[s] & solved && time[s] & solved && time[s] & solved \\
      \midrule
        Concurrency && 1828 & 40 && \textbf{1344} & 41 && 2414 & 39 && 1506 & \textbf{42} \\
        DeviceDrivers && 12156 & 241 && 12154 & 241 && 768 & \textbf{298} && \textbf{763} & \textbf{298} \\
        ECA && 20794 & \textbf{230} && \textbf{20388} & 228 && 38907 & 211 && 21606 & 211 \\
        ProductLines && 19571 & 276 && 19770 & 275 && 25361 & 272 && \textbf{11995} & \textbf{293} \\
        Sequentialized && 3710 & 44 && 3821 & 45 && 3200 & 43 && \textbf{1735} & \textbf{47} \\ \midrule
      \textbf{Summary} && 58061 & 831 && 57478 & 830 && 70652 & 863 && \textbf{37607} & \textbf{891} \\
      \bottomrule
    \end{tabular}
  \end{table*}

To evaluate effects of the state slicing, we have extended the
implementation of \symdivine. We have taken an advantage of the
internal modular structure of the tool, where each multi-state
component (control-flow with memory layout, explicit variable
representation and symbolic variable representation) is clearly
separated and has a well-defined interface. Therefore, a new
DataStore (\symdivine's terminology for a multi-state representation)
was implemented and integrated. In the extended version of \symdivine,
user can switch between the original SMTStore and our newly
implemented PartialStore, which uses the state slicing.

We also extended \symdivine with a possibility to cache both emptiness
and equality queries. The caching is performed before an internal
\symdivine formula representation is converted to the format specific
for the given \SMT solver. Therefore, caching is independent of the
DataStore used during the verification.

For the evaluation, we used benchmarks from the SV-COMP~\cite{SVCOMP}
in all categories except Termination. Note that not all categories are
included in results, as there are benchmarks that cannot be verified
by \symdivine due to the presence of dynamically allocated memory or
undefined intrinsic functions. We ran \symdivine in four different
configurations:
\begin{enumerate}
    \item SMTStore without caching,
    \item SMTStore with caching,
    \item PartialStore without caching,
    \item PartialStore with caching.
\end{enumerate}
Our test environment featured and Intel Xeon-5130 \textsc{cpu} (2.0~GHz) with
16~GB of physical \textsc{ram} and an Arch Linux distribution with
4.4.8-1-lts Linux kernel. \symdivine was built in the release
configuration using \textsc{gcc} 5.4 and used the \SMT solver Z3~\cite{Z3} in
version 4.4.1. Each verification task was restricted by a \textsc{cpu}
time limit of 10~minutes and by 7.5 GB memory limit. All SV-COMP
benchmarks were compiled with Clang 3.4 and \texttt{-O2} optimization level to
obtain the \llvm bitcode.

\begin{table*}[htb!]
    \centering
    \caption{Summary results showing effects of PartialStore and caching on
    number of \SMT solver calls.
    \emph{Equal checks} is the number of all issued equality checks. \emph{Syntactic equality} is a number of checks in which syntactically equivalent formulas were supplied. \emph{Cached} marks number of cache hits. \emph{Solver calls} is the number of checks, which required a query to an SMT solver.}\label{tab:checks}
    \bigskip
        \begin{tabular}{@{}lrrrrrrrrrr@{}}
        \toprule
        \multicolumn{1}{c}{}    & \multicolumn{1}{c}{Equal checks} &~& \multicolumn{2}{c}{Syntactic equality} &~& \multicolumn{2}{c}{Cached} &~& \multicolumn{2}{c}{Solver calls} \\ \midrule
          SMTStore                & 3\,110\,831                          &~& 2\,356\,034 & (76\%)                                &~& 0 & (0\%)                         &~& 754\,797 & (24\%)              \\
          PartialStore            & 81\,562\,863                         &~& 77\,132\,869 & (95\%)                               &~& 0 & (0\%)                         &~& 4\,429\,994 & (5\%)           \\
          SMTStore with cache     & 3\,620\,635                          &~& 2\,785\,076 & (77\%)                                &~& 136\,982 & (4\%)                  &~& 698\,577 & (19\%)             \\
          PartialStore with cache & 129\,798\,056                        &~& 122\,666\,491 & (95\%)                              &~& 6\,757\,205 &(5\%)               &~& 374\,360 & ($<$1\%)             \\ \bottomrule
        \end{tabular}
\end{table*}

The results are summarized in Table \ref{tab:summary} and in Figure
\ref{fig:all}. To show only relevant results, we have excluded
benchmarks with less than 10 states, because on such a small number of
states the performance is the same for each DataStore. These small
benchmarks are produced by the Clang optimizations; some benchmarks
are easy enough to be solved by static analysis or they contain
undefined behavior that Clang abuses to its benefit. This
significantly reduces the benchmark set, but it better corresponds to
the practice where compiler's optimizations are used frequently.

We examined the results of each category independently to see the
effects of slicing and caching on different types of input
programs. To be able to explain observed time differences, we also
measured the number of \SMT solver queries and number of cases in which the
query was avoided completely.

In general, multi-state equality is the most expensive operation
during the verification process and therefore we did not observe any
significant speedup caused only by slicing of emptiness queries. On
the other hand, slicing noticeably changes behavior of the equality
check, as can be seen in Table \ref{tab:checks}, which shows counts of
issued equality checks and whether they were solved by syntactic
equality, cache, or an \SMT solver. Slicing of multi-states increases
number of equality checks on our testing set at least by an order of
magnitude. This shows that issued queries usually contain a large
number of independent groups. If the caching is disabled, state
slicing also usually produces larger number of queries to an \SMT
solver. However, caching shows as highly effective on sliced
multi-states and therefore, by enabling it, we can save roughly half
of the queries issued to an \SMT solver compared to the original
version of \symdivine. Note that data for Table \ref{tab:checks} were
collected in the same manner as data for Table~\ref{tab:summary}, i.e. it
contains data from all benchmarks that were solved within the given
timeout. In particular, the number of equality checks is bigger for
more efficient configurations as they were able to decide more
benchmarks within the given time.

Overall, slicing comes with an overhead caused by two factors:
\begin{itemize}
    \item data dependencies have to be computed during the
      interpretation of the program and
    \item the number of potential solver queries is greater and each
      query has a constant overhead.
    \end{itemize}
However, the overhead can be outweighed by:
\begin{itemize}
\item simplicity of issued queries, as the solver does not have to
  deal with irrelevant variables;
\item syntactic equality, because the state slicing makes it easier to
  preserve syntactically equivalent formulas, which helps \symdivine
  as it can quickly recognize syntactically equal formulas; and
\item caching, as an equality query of sliced states is composed of many
  small, rarely changing queries to an \SMT solver.
\end{itemize}

The reasoning above lines up with results observed: configuration with
PartialStore without cache and the configuration with SMTStore with
cache rarely bring any speedup. The behavior was different only in categories
Concurrency and ECA.

In Concurrency set, the difference is caused by the simplicity of the
benchmarks and the presence of diamond-shapes in the state
space. Diamonds in the state space of simple benchmarks tend to
produce syntactically equivalent formulas and therefore \symdivine can
tremendously benefit from their detection. Only less than 2~\% of the
equality queries are handed out to an \SMT solver. Therefore, state slicing
increases the overhead and does not pay off. Benchmarks in the ECA
set are synthetically generated benchmarks with complicated data
dependencies and therefore multi-states of these benchmarks are rarely
sliceable.

On the other hand, our method significantly improves the performance
in the DeviceDriversLinux64 category, where it managed to save 93~\%
of the verification time and managed to verify 57 additional
benchmarks compared to the original configuration. The reason for that
is as follows: these benchmarks are closest to the real-world code and
contain a lot of code irrelevant to potential errors unlike benchmarks
in the other categories, which are usually reduced to the bare-bone of
the problem. In this environment, where multi-states are large and
program interpretation changes them only locally, slicing produces
large amount of syntactically equal queries and also allows for
a large number of cache hits.

In sum, slicing combined with caching performs best on large
benchmarks, as can be seen in Figure \ref{fig:all}, and can save
roughly 40~\% of the verification time. If we omit the category ECA, which
is in our point of view not-so-relevant for real usage, the time
savings go up to 60~\%. Note that state slicing does not cause any
significant memory usage increase as the amount of additional
information is rather small compared to the whole multi-state.

  \begin{figure*}[h!]
    \captionsetup[subfigure]{aboveskip=-18pt,belowskip=-8pt}
    \begin{subfigure}{.5\textwidth}
        \centering
        \input{graphics/Concurrency/dots.tex}
        \caption{Concurrency set}
        \label{fig:concurrency}
    \end{subfigure}%
    \begin{subfigure}{.5\textwidth}
        \centering
        \input{graphics/DeviceDriversLinux64/dots.tex}
        \caption{DeviceDriversLinux64 set}
        \label{fig:linux}
    \end{subfigure}
    \begin{subfigure}{.5\textwidth}
        \centering
        \input{graphics/ECA/dots.tex}
        \caption{ECA set}
        \label{fig:eca}
    \end{subfigure}%
    \begin{subfigure}{.5\textwidth}
        \centering
        \input{graphics/ProductLines/dots.tex}
        \caption{ProductLines set}
        \label{fig:product}
    \end{subfigure}
    \begin{subfigure}{.5\textwidth}
        \centering
        \input{graphics/Sequentialized/dots.tex}
        \caption{Sequentialized set}
        \label{fig:seq}
    \end{subfigure}%
    \begin{subfigure}{.5\textwidth}
        \centering
        \input{graphics/summary/dots.tex}
        \caption{Summary}
        \label{fig:summary}
    \end{subfigure}
    \bigskip
    \bigskip
    \caption{Comparison of verification time in seconds of basic
      SMTStore without cache ($x$ axis) and PartialStore with cache
      ($y$ axis). Diagonal lines denote double, equal and half the
      verification time. }
    \label{fig:all}
\end{figure*}
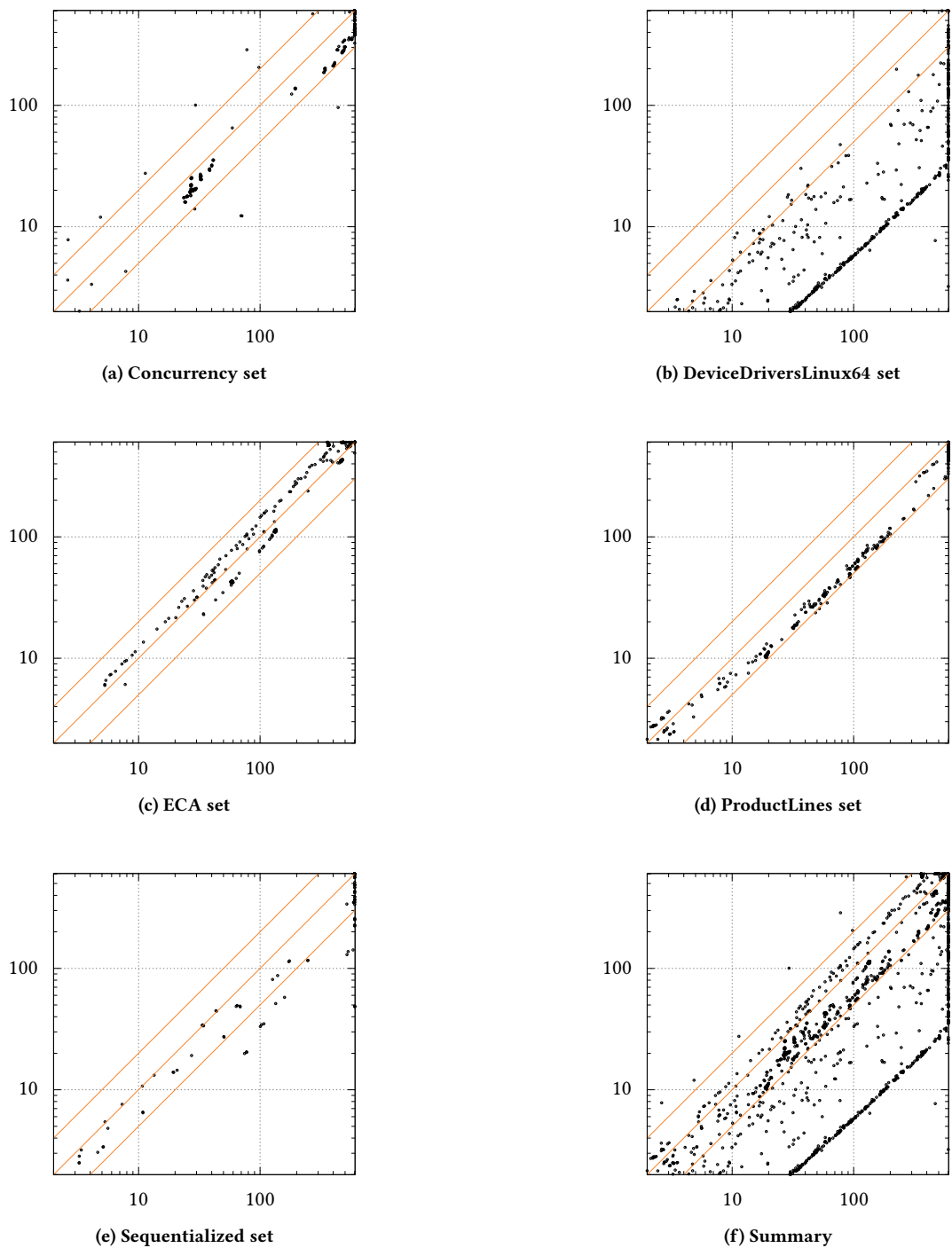

\section{Conclusion} \label{sec:conclusion}

We believe that state slicing combined with caching is a substantial
improvement to the Control-Explicit Data-Symbolic approach to
automated formal verification. Our experimental measurements confirm a
significant performance boost especially if the verified program is
similar to the real-world code. We still see some future work that
would be of much more technical and implementation nature. For
example, we can imagine that an incorporation of other formula
simplification methods, which are not present in Z3, could save more
verification time. The same goes for storing counter examples for
multi-state equality, which could be used to differentiate between two
multi-states without issuing a quantified \SMT query.

\bibliography{main}

\end{document}

%% file: graphics/Concurrency/dots.tex
\begingroup
  \makeatletter
  \providecommand\color[2][]{%
    \GenericError{(gnuplot) \space\space\space\@spaces}{%
      Package color not loaded in conjunction with
      terminal option `colourtext'%
    }{See the gnuplot documentation for explanation.%
    }{Either use 'blacktext' in gnuplot or load the package
      color.sty in LaTeX.}%
    \renewcommand\color[2][]{}%
  }%
  \providecommand\includegraphics[2][]{%
    \GenericError{(gnuplot) \space\space\space\@spaces}{%
      Package graphicx or graphics not loaded%
    }{See the gnuplot documentation for explanation.%
    }{The gnuplot epslatex terminal needs graphicx.sty or graphics.sty.}%
    \renewcommand\includegraphics[2][]{}%
  }%
  \providecommand\rotatebox[2]{#2}%
  \@ifundefined{ifGPcolor}{%
    \newif\ifGPcolor
    \GPcolortrue
  }{}%
  \@ifundefined{ifGPblacktext}{%
    \newif\ifGPblacktext
    \GPblacktextfalse
  }{}%
  \let\gplgaddtomacro\g@addto@macro
  \gdef\gplbacktext{}%
  \gdef\gplfronttext{}%
  \makeatother
  \ifGPblacktext
    \def\colorrgb#1{}%
    \def\colorgray#1{}%
  \else
    \ifGPcolor
      \def\colorrgb#1{\color[rgb]{#1}}%
      \def\colorgray#1{\color[gray]{#1}}%
      \expandafter\def\csname LTw\endcsname{\color{white}}%
      \expandafter\def\csname LTb\endcsname{\color{black}}%
      \expandafter\def\csname LTa\endcsname{\color{black}}%
      \expandafter\def\csname LT0\endcsname{\color[rgb]{1,0,0}}%
      \expandafter\def\csname LT1\endcsname{\color[rgb]{0,1,0}}%
      \expandafter\def\csname LT2\endcsname{\color[rgb]{0,0,1}}%
      \expandafter\def\csname LT3\endcsname{\color[rgb]{1,0,1}}%
      \expandafter\def\csname LT4\endcsname{\color[rgb]{0,1,1}}%
      \expandafter\def\csname LT5\endcsname{\color[rgb]{1,1,0}}%
      \expandafter\def\csname LT6\endcsname{\color[rgb]{0,0,0}}%
      \expandafter\def\csname LT7\endcsname{\color[rgb]{1,0.3,0}}%
      \expandafter\def\csname LT8\endcsname{\color[rgb]{0.5,0.5,0.5}}%
    \else
      \def\colorrgb#1{\color{black}}%
      \def\colorgray#1{\color[gray]{#1}}%
      \expandafter\def\csname LTw\endcsname{\color{white}}%
      \expandafter\def\csname LTb\endcsname{\color{black}}%
      \expandafter\def\csname LTa\endcsname{\color{black}}%
      \expandafter\def\csname LT0\endcsname{\color{black}}%
      \expandafter\def\csname LT1\endcsname{\color{black}}%
      \expandafter\def\csname LT2\endcsname{\color{black}}%
      \expandafter\def\csname LT3\endcsname{\color{black}}%
      \expandafter\def\csname LT4\endcsname{\color{black}}%
      \expandafter\def\csname LT5\endcsname{\color{black}}%
      \expandafter\def\csname LT6\endcsname{\color{black}}%
      \expandafter\def\csname LT7\endcsname{\color{black}}%
      \expandafter\def\csname LT8\endcsname{\color{black}}%
    \fi
  \fi
    \setlength{\unitlength}{0.0500bp}%
    \ifx\gptboxheight\undefined%
      \newlength{\gptboxheight}%
      \newlength{\gptboxwidth}%
      \newsavebox{\gptboxtext}%
    \fi%
    \setlength{\fboxrule}{0.5pt}%
    \setlength{\fboxsep}{1pt}%
\begin{picture}(3684.00,3968.00)%
    \gplgaddtomacro\gplbacktext{%
      \csname LTb\endcsname%
      \put(594,1513){\makebox(0,0)[r]{\strut{}10}}%
      \csname LTb\endcsname%
      \put(594,2545){\makebox(0,0)[r]{\strut{}100}}%
      \csname LTb\endcsname%
      \put(1448,571){\makebox(0,0){\strut{}$10$}}%
      \csname LTb\endcsname%
      \put(2480,571){\makebox(0,0){\strut{}$100$}}%
    }%
    \gplgaddtomacro\gplfronttext{%
    }%
    \gplbacktext
    \put(0,0){\includegraphics{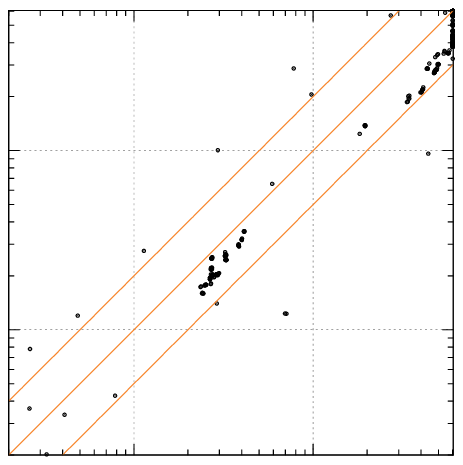}}%
    \gplfronttext
  \end{picture}%
\endgroup

%% file: graphics/DeviceDriversLinux64/dots.tex
\begingroup
  \makeatletter
  \providecommand\color[2][]{%
    \GenericError{(gnuplot) \space\space\space\@spaces}{%
      Package color not loaded in conjunction with
      terminal option `colourtext'%
    }{See the gnuplot documentation for explanation.%
    }{Either use 'blacktext' in gnuplot or load the package
      color.sty in LaTeX.}%
    \renewcommand\color[2][]{}%
  }%
  \providecommand\includegraphics[2][]{%
    \GenericError{(gnuplot) \space\space\space\@spaces}{%
      Package graphicx or graphics not loaded%
    }{See the gnuplot documentation for explanation.%
    }{The gnuplot epslatex terminal needs graphicx.sty or graphics.sty.}%
    \renewcommand\includegraphics[2][]{}%
  }%
  \providecommand\rotatebox[2]{#2}%
  \@ifundefined{ifGPcolor}{%
    \newif\ifGPcolor
    \GPcolortrue
  }{}%
  \@ifundefined{ifGPblacktext}{%
    \newif\ifGPblacktext
    \GPblacktextfalse
  }{}%
  \let\gplgaddtomacro\g@addto@macro
  \gdef\gplbacktext{}%
  \gdef\gplfronttext{}%
  \makeatother
  \ifGPblacktext
    \def\colorrgb#1{}%
    \def\colorgray#1{}%
  \else
    \ifGPcolor
      \def\colorrgb#1{\color[rgb]{#1}}%
      \def\colorgray#1{\color[gray]{#1}}%
      \expandafter\def\csname LTw\endcsname{\color{white}}%
      \expandafter\def\csname LTb\endcsname{\color{black}}%
      \expandafter\def\csname LTa\endcsname{\color{black}}%
      \expandafter\def\csname LT0\endcsname{\color[rgb]{1,0,0}}%
      \expandafter\def\csname LT1\endcsname{\color[rgb]{0,1,0}}%
      \expandafter\def\csname LT2\endcsname{\color[rgb]{0,0,1}}%
      \expandafter\def\csname LT3\endcsname{\color[rgb]{1,0,1}}%
      \expandafter\def\csname LT4\endcsname{\color[rgb]{0,1,1}}%
      \expandafter\def\csname LT5\endcsname{\color[rgb]{1,1,0}}%
      \expandafter\def\csname LT6\endcsname{\color[rgb]{0,0,0}}%
      \expandafter\def\csname LT7\endcsname{\color[rgb]{1,0.3,0}}%
      \expandafter\def\csname LT8\endcsname{\color[rgb]{0.5,0.5,0.5}}%
    \else
      \def\colorrgb#1{\color{black}}%
      \def\colorgray#1{\color[gray]{#1}}%
      \expandafter\def\csname LTw\endcsname{\color{white}}%
      \expandafter\def\csname LTb\endcsname{\color{black}}%
      \expandafter\def\csname LTa\endcsname{\color{black}}%
      \expandafter\def\csname LT0\endcsname{\color{black}}%
      \expandafter\def\csname LT1\endcsname{\color{black}}%
      \expandafter\def\csname LT2\endcsname{\color{black}}%
      \expandafter\def\csname LT3\endcsname{\color{black}}%
      \expandafter\def\csname LT4\endcsname{\color{black}}%
      \expandafter\def\csname LT5\endcsname{\color{black}}%
      \expandafter\def\csname LT6\endcsname{\color{black}}%
      \expandafter\def\csname LT7\endcsname{\color{black}}%
      \expandafter\def\csname LT8\endcsname{\color{black}}%
    \fi
  \fi
    \setlength{\unitlength}{0.0500bp}%
    \ifx\gptboxheight\undefined%
      \newlength{\gptboxheight}%
      \newlength{\gptboxwidth}%
      \newsavebox{\gptboxtext}%
    \fi%
    \setlength{\fboxrule}{0.5pt}%
    \setlength{\fboxsep}{1pt}%
\begin{picture}(3684.00,3968.00)%
    \gplgaddtomacro\gplbacktext{%
      \csname LTb\endcsname%
      \put(594,1513){\makebox(0,0)[r]{\strut{}10}}%
      \csname LTb\endcsname%
      \put(594,2545){\makebox(0,0)[r]{\strut{}100}}%
      \csname LTb\endcsname%
      \put(1448,571){\makebox(0,0){\strut{}$10$}}%
      \csname LTb\endcsname%
      \put(2480,571){\makebox(0,0){\strut{}$100$}}%
    }%
    \gplgaddtomacro\gplfronttext{%
    }%
    \gplbacktext
    \put(0,0){\includegraphics{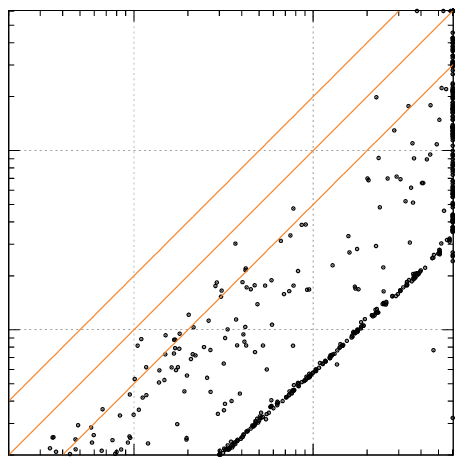}}%
    \gplfronttext
  \end{picture}%
\endgroup

%% file: graphics/ECA/dots.tex
\begingroup
  \makeatletter
  \providecommand\color[2][]{%
    \GenericError{(gnuplot) \space\space\space\@spaces}{%
      Package color not loaded in conjunction with
      terminal option `colourtext'%
    }{See the gnuplot documentation for explanation.%
    }{Either use 'blacktext' in gnuplot or load the package
      color.sty in LaTeX.}%
    \renewcommand\color[2][]{}%
  }%
  \providecommand\includegraphics[2][]{%
    \GenericError{(gnuplot) \space\space\space\@spaces}{%
      Package graphicx or graphics not loaded%
    }{See the gnuplot documentation for explanation.%
    }{The gnuplot epslatex terminal needs graphicx.sty or graphics.sty.}%
    \renewcommand\includegraphics[2][]{}%
  }%
  \providecommand\rotatebox[2]{#2}%
  \@ifundefined{ifGPcolor}{%
    \newif\ifGPcolor
    \GPcolortrue
  }{}%
  \@ifundefined{ifGPblacktext}{%
    \newif\ifGPblacktext
    \GPblacktextfalse
  }{}%
  \let\gplgaddtomacro\g@addto@macro
  \gdef\gplbacktext{}%
  \gdef\gplfronttext{}%
  \makeatother
  \ifGPblacktext
    \def\colorrgb#1{}%
    \def\colorgray#1{}%
  \else
    \ifGPcolor
      \def\colorrgb#1{\color[rgb]{#1}}%
      \def\colorgray#1{\color[gray]{#1}}%
      \expandafter\def\csname LTw\endcsname{\color{white}}%
      \expandafter\def\csname LTb\endcsname{\color{black}}%
      \expandafter\def\csname LTa\endcsname{\color{black}}%
      \expandafter\def\csname LT0\endcsname{\color[rgb]{1,0,0}}%
      \expandafter\def\csname LT1\endcsname{\color[rgb]{0,1,0}}%
      \expandafter\def\csname LT2\endcsname{\color[rgb]{0,0,1}}%
      \expandafter\def\csname LT3\endcsname{\color[rgb]{1,0,1}}%
      \expandafter\def\csname LT4\endcsname{\color[rgb]{0,1,1}}%
      \expandafter\def\csname LT5\endcsname{\color[rgb]{1,1,0}}%
      \expandafter\def\csname LT6\endcsname{\color[rgb]{0,0,0}}%
      \expandafter\def\csname LT7\endcsname{\color[rgb]{1,0.3,0}}%
      \expandafter\def\csname LT8\endcsname{\color[rgb]{0.5,0.5,0.5}}%
    \else
      \def\colorrgb#1{\color{black}}%
      \def\colorgray#1{\color[gray]{#1}}%
      \expandafter\def\csname LTw\endcsname{\color{white}}%
      \expandafter\def\csname LTb\endcsname{\color{black}}%
      \expandafter\def\csname LTa\endcsname{\color{black}}%
      \expandafter\def\csname LT0\endcsname{\color{black}}%
      \expandafter\def\csname LT1\endcsname{\color{black}}%
      \expandafter\def\csname LT2\endcsname{\color{black}}%
      \expandafter\def\csname LT3\endcsname{\color{black}}%
      \expandafter\def\csname LT4\endcsname{\color{black}}%
      \expandafter\def\csname LT5\endcsname{\color{black}}%
      \expandafter\def\csname LT6\endcsname{\color{black}}%
      \expandafter\def\csname LT7\endcsname{\color{black}}%
      \expandafter\def\csname LT8\endcsname{\color{black}}%
    \fi
  \fi
    \setlength{\unitlength}{0.0500bp}%
    \ifx\gptboxheight\undefined%
      \newlength{\gptboxheight}%
      \newlength{\gptboxwidth}%
      \newsavebox{\gptboxtext}%
    \fi%
    \setlength{\fboxrule}{0.5pt}%
    \setlength{\fboxsep}{1pt}%
\begin{picture}(3684.00,3968.00)%
    \gplgaddtomacro\gplbacktext{%
      \csname LTb\endcsname%
      \put(594,1513){\makebox(0,0)[r]{\strut{}10}}%
      \csname LTb\endcsname%
      \put(594,2545){\makebox(0,0)[r]{\strut{}100}}%
      \csname LTb\endcsname%
      \put(1448,571){\makebox(0,0){\strut{}$10$}}%
      \csname LTb\endcsname%
      \put(2480,571){\makebox(0,0){\strut{}$100$}}%
    }%
    \gplgaddtomacro\gplfronttext{%
    }%
    \gplbacktext
    \put(0,0){\includegraphics{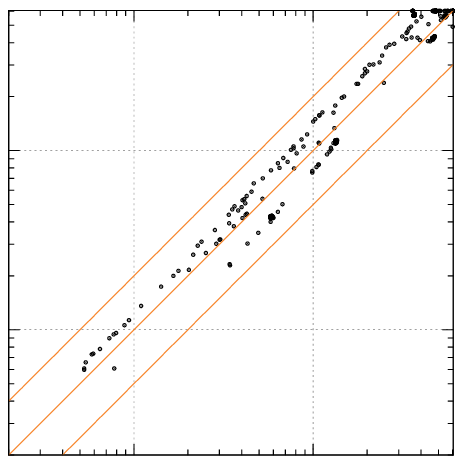}}%
    \gplfronttext
  \end{picture}%
\endgroup

%% file: graphics/ProductLines/dots.tex
\begingroup
  \makeatletter
  \providecommand\color[2][]{%
    \GenericError{(gnuplot) \space\space\space\@spaces}{%
      Package color not loaded in conjunction with
      terminal option `colourtext'%
    }{See the gnuplot documentation for explanation.%
    }{Either use 'blacktext' in gnuplot or load the package
      color.sty in LaTeX.}%
    \renewcommand\color[2][]{}%
  }%
  \providecommand\includegraphics[2][]{%
    \GenericError{(gnuplot) \space\space\space\@spaces}{%
      Package graphicx or graphics not loaded%
    }{See the gnuplot documentation for explanation.%
    }{The gnuplot epslatex terminal needs graphicx.sty or graphics.sty.}%
    \renewcommand\includegraphics[2][]{}%
  }%
  \providecommand\rotatebox[2]{#2}%
  \@ifundefined{ifGPcolor}{%
    \newif\ifGPcolor
    \GPcolortrue
  }{}%
  \@ifundefined{ifGPblacktext}{%
    \newif\ifGPblacktext
    \GPblacktextfalse
  }{}%
  \let\gplgaddtomacro\g@addto@macro
  \gdef\gplbacktext{}%
  \gdef\gplfronttext{}%
  \makeatother
  \ifGPblacktext
    \def\colorrgb#1{}%
    \def\colorgray#1{}%
  \else
    \ifGPcolor
      \def\colorrgb#1{\color[rgb]{#1}}%
      \def\colorgray#1{\color[gray]{#1}}%
      \expandafter\def\csname LTw\endcsname{\color{white}}%
      \expandafter\def\csname LTb\endcsname{\color{black}}%
      \expandafter\def\csname LTa\endcsname{\color{black}}%
      \expandafter\def\csname LT0\endcsname{\color[rgb]{1,0,0}}%
      \expandafter\def\csname LT1\endcsname{\color[rgb]{0,1,0}}%
      \expandafter\def\csname LT2\endcsname{\color[rgb]{0,0,1}}%
      \expandafter\def\csname LT3\endcsname{\color[rgb]{1,0,1}}%
      \expandafter\def\csname LT4\endcsname{\color[rgb]{0,1,1}}%
      \expandafter\def\csname LT5\endcsname{\color[rgb]{1,1,0}}%
      \expandafter\def\csname LT6\endcsname{\color[rgb]{0,0,0}}%
      \expandafter\def\csname LT7\endcsname{\color[rgb]{1,0.3,0}}%
      \expandafter\def\csname LT8\endcsname{\color[rgb]{0.5,0.5,0.5}}%
    \else
      \def\colorrgb#1{\color{black}}%
      \def\colorgray#1{\color[gray]{#1}}%
      \expandafter\def\csname LTw\endcsname{\color{white}}%
      \expandafter\def\csname LTb\endcsname{\color{black}}%
      \expandafter\def\csname LTa\endcsname{\color{black}}%
      \expandafter\def\csname LT0\endcsname{\color{black}}%
      \expandafter\def\csname LT1\endcsname{\color{black}}%
      \expandafter\def\csname LT2\endcsname{\color{black}}%
      \expandafter\def\csname LT3\endcsname{\color{black}}%
      \expandafter\def\csname LT4\endcsname{\color{black}}%
      \expandafter\def\csname LT5\endcsname{\color{black}}%
      \expandafter\def\csname LT6\endcsname{\color{black}}%
      \expandafter\def\csname LT7\endcsname{\color{black}}%
      \expandafter\def\csname LT8\endcsname{\color{black}}%
    \fi
  \fi
    \setlength{\unitlength}{0.0500bp}%
    \ifx\gptboxheight\undefined%
      \newlength{\gptboxheight}%
      \newlength{\gptboxwidth}%
      \newsavebox{\gptboxtext}%
    \fi%
    \setlength{\fboxrule}{0.5pt}%
    \setlength{\fboxsep}{1pt}%
\begin{picture}(3684.00,3968.00)%
    \gplgaddtomacro\gplbacktext{%
      \csname LTb\endcsname%
      \put(594,1513){\makebox(0,0)[r]{\strut{}10}}%
      \csname LTb\endcsname%
      \put(594,2545){\makebox(0,0)[r]{\strut{}100}}%
      \csname LTb\endcsname%
      \put(1448,571){\makebox(0,0){\strut{}$10$}}%
      \csname LTb\endcsname%
      \put(2480,571){\makebox(0,0){\strut{}$100$}}%
    }%
    \gplgaddtomacro\gplfronttext{%
    }%
    \gplbacktext
    \put(0,0){\includegraphics{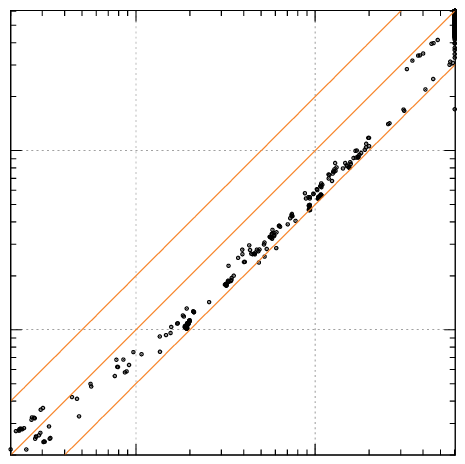}}%
    \gplfronttext
  \end{picture}%
\endgroup

%% file: graphics/Sequentialized/dots.tex
\begingroup
  \makeatletter
  \providecommand\color[2][]{%
    \GenericError{(gnuplot) \space\space\space\@spaces}{%
      Package color not loaded in conjunction with
      terminal option `colourtext'%
    }{See the gnuplot documentation for explanation.%
    }{Either use 'blacktext' in gnuplot or load the package
      color.sty in LaTeX.}%
    \renewcommand\color[2][]{}%
  }%
  \providecommand\includegraphics[2][]{%
    \GenericError{(gnuplot) \space\space\space\@spaces}{%
      Package graphicx or graphics not loaded%
    }{See the gnuplot documentation for explanation.%
    }{The gnuplot epslatex terminal needs graphicx.sty or graphics.sty.}%
    \renewcommand\includegraphics[2][]{}%
  }%
  \providecommand\rotatebox[2]{#2}%
  \@ifundefined{ifGPcolor}{%
    \newif\ifGPcolor
    \GPcolortrue
  }{}%
  \@ifundefined{ifGPblacktext}{%
    \newif\ifGPblacktext
    \GPblacktextfalse
  }{}%
  \let\gplgaddtomacro\g@addto@macro
  \gdef\gplbacktext{}%
  \gdef\gplfronttext{}%
  \makeatother
  \ifGPblacktext
    \def\colorrgb#1{}%
    \def\colorgray#1{}%
  \else
    \ifGPcolor
      \def\colorrgb#1{\color[rgb]{#1}}%
      \def\colorgray#1{\color[gray]{#1}}%
      \expandafter\def\csname LTw\endcsname{\color{white}}%
      \expandafter\def\csname LTb\endcsname{\color{black}}%
      \expandafter\def\csname LTa\endcsname{\color{black}}%
      \expandafter\def\csname LT0\endcsname{\color[rgb]{1,0,0}}%
      \expandafter\def\csname LT1\endcsname{\color[rgb]{0,1,0}}%
      \expandafter\def\csname LT2\endcsname{\color[rgb]{0,0,1}}%
      \expandafter\def\csname LT3\endcsname{\color[rgb]{1,0,1}}%
      \expandafter\def\csname LT4\endcsname{\color[rgb]{0,1,1}}%
      \expandafter\def\csname LT5\endcsname{\color[rgb]{1,1,0}}%
      \expandafter\def\csname LT6\endcsname{\color[rgb]{0,0,0}}%
      \expandafter\def\csname LT7\endcsname{\color[rgb]{1,0.3,0}}%
      \expandafter\def\csname LT8\endcsname{\color[rgb]{0.5,0.5,0.5}}%
    \else
      \def\colorrgb#1{\color{black}}%
      \def\colorgray#1{\color[gray]{#1}}%
      \expandafter\def\csname LTw\endcsname{\color{white}}%
      \expandafter\def\csname LTb\endcsname{\color{black}}%
      \expandafter\def\csname LTa\endcsname{\color{black}}%
      \expandafter\def\csname LT0\endcsname{\color{black}}%
      \expandafter\def\csname LT1\endcsname{\color{black}}%
      \expandafter\def\csname LT2\endcsname{\color{black}}%
      \expandafter\def\csname LT3\endcsname{\color{black}}%
      \expandafter\def\csname LT4\endcsname{\color{black}}%
      \expandafter\def\csname LT5\endcsname{\color{black}}%
      \expandafter\def\csname LT6\endcsname{\color{black}}%
      \expandafter\def\csname LT7\endcsname{\color{black}}%
      \expandafter\def\csname LT8\endcsname{\color{black}}%
    \fi
  \fi
    \setlength{\unitlength}{0.0500bp}%
    \ifx\gptboxheight\undefined%
      \newlength{\gptboxheight}%
      \newlength{\gptboxwidth}%
      \newsavebox{\gptboxtext}%
    \fi%
    \setlength{\fboxrule}{0.5pt}%
    \setlength{\fboxsep}{1pt}%
\begin{picture}(3684.00,3968.00)%
    \gplgaddtomacro\gplbacktext{%
      \csname LTb\endcsname%
      \put(594,1513){\makebox(0,0)[r]{\strut{}10}}%
      \csname LTb\endcsname%
      \put(594,2545){\makebox(0,0)[r]{\strut{}100}}%
      \csname LTb\endcsname%
      \put(1448,571){\makebox(0,0){\strut{}$10$}}%
      \csname LTb\endcsname%
      \put(2480,571){\makebox(0,0){\strut{}$100$}}%
    }%
    \gplgaddtomacro\gplfronttext{%
    }%
    \gplbacktext
    \put(0,0){\includegraphics{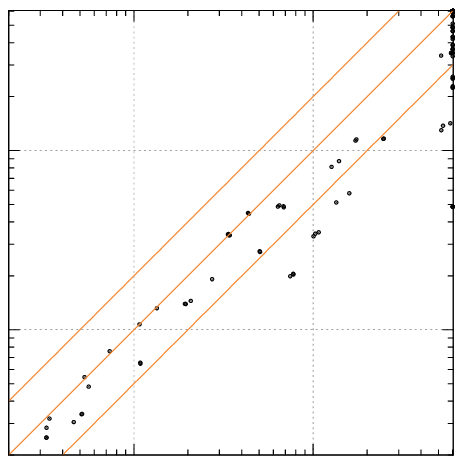}}%
    \gplfronttext
  \end{picture}%
\endgroup

%% file: graphics/summary/dots.tex
\begingroup
  \makeatletter
  \providecommand\color[2][]{%
    \GenericError{(gnuplot) \space\space\space\@spaces}{%
      Package color not loaded in conjunction with
      terminal option `colourtext'%
    }{See the gnuplot documentation for explanation.%
    }{Either use 'blacktext' in gnuplot or load the package
      color.sty in LaTeX.}%
    \renewcommand\color[2][]{}%
  }%
  \providecommand\includegraphics[2][]{%
    \GenericError{(gnuplot) \space\space\space\@spaces}{%
      Package graphicx or graphics not loaded%
    }{See the gnuplot documentation for explanation.%
    }{The gnuplot epslatex terminal needs graphicx.sty or graphics.sty.}%
    \renewcommand\includegraphics[2][]{}%
  }%
  \providecommand\rotatebox[2]{#2}%
  \@ifundefined{ifGPcolor}{%
    \newif\ifGPcolor
    \GPcolortrue
  }{}%
  \@ifundefined{ifGPblacktext}{%
    \newif\ifGPblacktext
    \GPblacktextfalse
  }{}%
  \let\gplgaddtomacro\g@addto@macro
  \gdef\gplbacktext{}%
  \gdef\gplfronttext{}%
  \makeatother
  \ifGPblacktext
    \def\colorrgb#1{}%
    \def\colorgray#1{}%
  \else
    \ifGPcolor
      \def\colorrgb#1{\color[rgb]{#1}}%
      \def\colorgray#1{\color[gray]{#1}}%
      \expandafter\def\csname LTw\endcsname{\color{white}}%
      \expandafter\def\csname LTb\endcsname{\color{black}}%
      \expandafter\def\csname LTa\endcsname{\color{black}}%
      \expandafter\def\csname LT0\endcsname{\color[rgb]{1,0,0}}%
      \expandafter\def\csname LT1\endcsname{\color[rgb]{0,1,0}}%
      \expandafter\def\csname LT2\endcsname{\color[rgb]{0,0,1}}%
      \expandafter\def\csname LT3\endcsname{\color[rgb]{1,0,1}}%
      \expandafter\def\csname LT4\endcsname{\color[rgb]{0,1,1}}%
      \expandafter\def\csname LT5\endcsname{\color[rgb]{1,1,0}}%
      \expandafter\def\csname LT6\endcsname{\color[rgb]{0,0,0}}%
      \expandafter\def\csname LT7\endcsname{\color[rgb]{1,0.3,0}}%
      \expandafter\def\csname LT8\endcsname{\color[rgb]{0.5,0.5,0.5}}%
    \else
      \def\colorrgb#1{\color{black}}%
      \def\colorgray#1{\color[gray]{#1}}%
      \expandafter\def\csname LTw\endcsname{\color{white}}%
      \expandafter\def\csname LTb\endcsname{\color{black}}%
      \expandafter\def\csname LTa\endcsname{\color{black}}%
      \expandafter\def\csname LT0\endcsname{\color{black}}%
      \expandafter\def\csname LT1\endcsname{\color{black}}%
      \expandafter\def\csname LT2\endcsname{\color{black}}%
      \expandafter\def\csname LT3\endcsname{\color{black}}%
      \expandafter\def\csname LT4\endcsname{\color{black}}%
      \expandafter\def\csname LT5\endcsname{\color{black}}%
      \expandafter\def\csname LT6\endcsname{\color{black}}%
      \expandafter\def\csname LT7\endcsname{\color{black}}%
      \expandafter\def\csname LT8\endcsname{\color{black}}%
    \fi
  \fi
    \setlength{\unitlength}{0.0500bp}%
    \ifx\gptboxheight\undefined%
      \newlength{\gptboxheight}%
      \newlength{\gptboxwidth}%
      \newsavebox{\gptboxtext}%
    \fi%
    \setlength{\fboxrule}{0.5pt}%
    \setlength{\fboxsep}{1pt}%
\begin{picture}(3684.00,3968.00)%
    \gplgaddtomacro\gplbacktext{%
      \csname LTb\endcsname%
      \put(594,1513){\makebox(0,0)[r]{\strut{}10}}%
      \csname LTb\endcsname%
      \put(594,2545){\makebox(0,0)[r]{\strut{}100}}%
      \csname LTb\endcsname%
      \put(1448,571){\makebox(0,0){\strut{}$10$}}%
      \csname LTb\endcsname%
      \put(2480,571){\makebox(0,0){\strut{}$100$}}%
    }%
    \gplgaddtomacro\gplfronttext{%
    }%
    \gplbacktext
    \put(0,0){\includegraphics{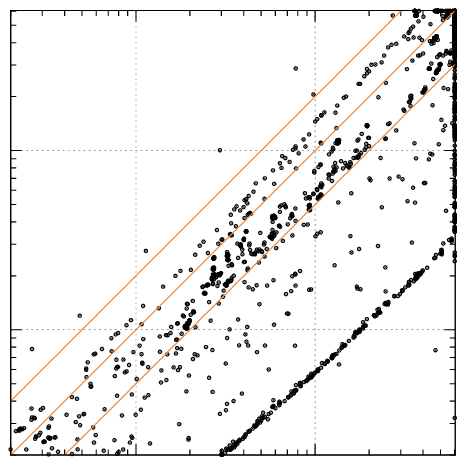}}%
    \gplfronttext
  \end{picture}%
\endgroup